\documentclass[12pt,a4paper]{article}

\usepackage[top=2.5cm,bottom=2.5cm,left=2.2cm,right=2.2cm]{geometry}
\usepackage{latexsym,amsmath,amsfonts,amssymb,amsthm,amscd,mathrsfs}

\newcommand{\be}{\begin{equation}}
\newcommand{\ee}{\end{equation}}
\newcommand{\bea}{\begin{eqnarray}}
\newcommand{\eea}{\end{eqnarray}}
\newcommand{\ba}{\begin{aligned}}
\newcommand{\ea}{\end{aligned}}
\numberwithin{equation}{section}
\newcounter{thmcounter}
\numberwithin{thmcounter}{section}

\theoremstyle{definition}

\theoremstyle{plain}

\newtheorem{proposition}[thmcounter]{Proposition}
\newtheorem{theorem}[thmcounter]{Theorem}

\def\BC{\mathrm{BC}}                        %
\def\1{{\boldsymbol 1}}                     %
\def\0{{\boldsymbol 0}}                     %
\def\CC{{\mathbb C}}                    %
\def\RR{{\mathbb R}}                    %
\def\TT{\mathbb{T}}                          %
\def\XX{{\mathbb X}}                    %
\def\YY{{\mathbb Y}}
\def\cD{{\mathcal D}}                       %
\def\cG{{\mathcal G}}                       %
\def\cH{{\mathcal H}}                       %
\def\cM{{\mathcal M}}                       %
\def\cR{{\mathcal R}}                       %
\def\cW{{\mathcal W}}                      %
\def\g{{\mathfrak g}}
 \def\b{{\mathfrak b}}
 \def\k{{\mathfrak k}}
\def\SL{{\rm SL}}                           %
\def\GL{{\rm GL}}                           %
\def\SU{{\rm SU}}                           %
\def\U{{\rm U}}                                     %

\def\red{\mathrm{red}}                    %
\def\tr{\mathrm{tr\,}}                        %
\def\diag{\mathrm{diag}}                    %
\def\ri{{\rm i}}                            %
\def\tw {\tilde w}
\def\Inn {I}
\def\Herm{Herm}
\def\qHerm{qHerm}
\def\sl{{\rm sl}}                           %

\def\G{\Gamma}
\def\S{\Sigma}
\def\Lm{{\Lambda}}
\def\q{{\lambda}}
\newcommand{\bpm}{\begin{pmatrix}}
\newcommand{\epm}{\end{pmatrix}}
\def\cinf{C^\infty}
\def\half{\textstyle{\frac12}}
\def\dt {\left.\frac{d}{dt}\right|_{t=0}}
\DeclareMathOperator*{\rez}{Res}

\begin{document}

\begin{center}
{\large\bf
The action-angle dual of an integrable Hamiltonian system of
Ruijsenaars--Schneider--van Diejen type}
\end{center}

\medskip
\begin{center}
L.~Feh\'er${}^{a,b}$ and I.~Marshall${}^c$\\

\bigskip
${}^a$Department of Theoretical Physics, University of Szeged\\
Tisza Lajos krt 84-86, H-6720 Szeged, Hungary\\
e-mail:  lfeher@physx.u-szeged.hu

\medskip
${}^b$Department of Theoretical Physics, WIGNER RCP, RMKI\\
H-1525 Budapest, P.O.B.~49, Hungary

\medskip
${}^c$
Faculty of Mathematics, Higher School of Economics\\
 Ulitsa Vavilova 7, Moscow, Russia\\
 e-mail: imarshall@hse.ru

\end{center}

\medskip
\begin{abstract}
Integrable deformations of the hyperbolic and trigonometric $\BC_n$ Sutherland models
were recently derived via Hamiltonian reduction of certain free systems
on the Heisenberg doubles of
$\SU(n,n)$ and $\SU(2n)$, respectively.
As a step towards constructing action-angle variables for these models, we here
 apply the same reduction to a different free system on the double of $\SU(2n)$ and
 thereby obtain a novel integrable many-body model of Ruijsenaars--Schneider--van Diejen type
that is in action-angle duality with the respective deformed Sutherland model.
\end{abstract}

\newpage

\newpage
\section{Introduction}
\label{sec:0}

The study of integrable many-body models of Calogero--Moser--Sutherland type
began with the  seminal  papers \cite{Cal,Suth,Mos} and has since been enriched by several
contributions, including notably the generalization to arbitrary root systems
by Olshanetsky and Perelomov \cite{OP1}, and the discovery of relativistic deformations
by Ruijsenaars and Schneider \cite{RS86} developed further by van Diejen \cite{vD1} and others.
These models are ubiquitous in physical applications and are connected to
important fields of mathematics; see the reviews
\cite{N,Banff,vDV,SuthR,PolR, EtiR}.

At the classical level,   these
models exhibit intriguing action-angle duality relations \cite{SR88,RIMS95}.
The duality of two integrable many-body models means that the position variables
of one model serve also as the action variables of the other one, and vice versa.
The pioneering work of Ruijsenaars \cite{SR88,RIMS95} relied on direct methods, building on
and greatly generalizing a procedure
that  had appeared in the  Hamiltonian reduction treatment of the simplest example
\cite{KKS}.
By now it has become  widely known \cite{JHEP,FK1}  that several dual pairs of models
arise  by applying
Hamiltonian reduction to suitable pairs of ``free systems'' on a higher dimensional
master phase space, and, whenever available,  this interpretation provides a powerful tool
for the analysis of the dual pairs.
The term  \textit{free system} is a loose one:   a free Hamiltonian induces a complete flow,
which often can be written down explicitly,
and  participates in a large Abelian Poisson algebra invariant
under a group of symmetries.

The goal of this paper is to exhibit action-angle duality
for an integrable Ruijsenaars--Schneider--van Diejen (RSvD) type model derived
 recently \cite{M,FG} by Hamiltonian reduction of the Heisenberg double \cite{STS}
of the Poisson Lie group $\SU(2n)$.
The model in question has three free parameters and is a deformation of the trigonometric
$\BC_n$ Sutherland model.  It  can be viewed also as a singular limit of a specialization
of the five-parameter deformation due to van Diejen \cite{vD1}.
Its derivation \cite{FG} closely followed the  analogous reduction \cite{M}
of the Heisenberg double of $\SU(n,n)$.
The papers \cite{M, FG} (see also \cite{FG1,M1}) applied Poisson-Lie analogues of the reduction
of the cotangent bundle of $\SU(n,n)$ that yields the hyperbolic $\BC_n$
Sutherland model with three arbitrary coupling constants \cite{FP}.   Other relevant predecessors
of the present work are the paper of Pusztai \cite{P},
 where the action-angle
dual of the hyperbolic $\BC_n$ Sutherland model was constructed by reduction
of $T^* \SU(n,n)$,  and its adaptation \cite{FG0} to the trigonometric case.

A key ingredient of every Hamiltonian reduction is the choice of
symmetry group, which in the above examples is the group $K_+ \times K_+$ with
$K_+ = \SU(n,n)\cap \SU(2n)$.
The pertinent
Heisenberg doubles carry two natural $(K_+\times K_+)$-symmetric free systems,
and  the previous works investigated reductions of those systems corresponding to
geodesic motion.
In the present article, we analyse the same reduction of the Heisenberg double of
$\SU(2n)$ as in \cite{FG}, but develop a new model of the reduced phase space,
wherein it is the other free system whose reduction admits a many-body interpretation.
In combination with the  earlier  results,  this allows us
to establish action-angle duality between the model treated in \cite{FG}
and the  many-body model that we obtain here.
The Hamiltonians of this pair of RSvD type
models are given in equations (\ref{distilham}) and (\ref{Z9}) below, and their duality with
one another is discussed in Section \ref{discussion}.

As in  \cite{M}, we  adopt the modest aim of
finding a model for a dense open subset of the reduced phase space.
  Full description of the complete reduced
 phase space will be  reported in another publication. It is worth emphasizing
that the investigation of the global structure of the
phase space emerging from Hamiltonian reduction can be
a source of rich and surprising results. An example is the study by Wilson \cite{W}
of the adelic Grassmannian related to the complexified rational Calogero--Moser system, which opened up
interesting connections
between commuting KP flows and bispectral operators.
It is also worth noting that a global description is necessary in order to obtain
complete flows after reduction, and this can be turned around
 to construct natural regularizations of several systems  with singularities.

  Section 2 is devoted to preparations.
 The two families of free Hamiltonians
 and their Hamiltonian vector fields are characterized in Proposition 2.1,
 and the reduction of interest is defined in Subsection 2.3.
 Our main new results are summarized  by Theorem 3.2 and Theorem 3.3
 in Section 3.  These describe Darboux coordinates on the reduced phase space
  in which the simplest reduced Hamiltonian descending from the
 second free system acquires an RSvD form.
 Proposition 3.1 formulates a  technical result that plays a key role in our analysis.
 In Section 4 we exhibit the action-angle duality
 between the reduced system derived in \cite{FG} and the one
 treated here for the first time. Finally, in Appendix A  the
 rational limit is presented of our RSvD type Hamiltonian (\ref{distilham}).

\section{Preliminaries}

We here collect the necessary definitions and background results that will be used later.
Most of these results are fairly standard and can be found  in many sources
(see e.g.~\cite{STSlectures}).

{\subsection{Group actions and invariants}}

 Our master phase space will be $\cM=\SL(2n,\CC)$,  treated as a real manifold.
 Let $K=\SU(2n)$, and $B$ the group consisting of the
upper triangular elements of $\SL(2n,\CC)$ with real, positive diagonal entries.
We shall use the notation $B_n$ for the analogous subgroup of $\GL(n,\CC)$.
 By the procedure of  Gram--Schmidt orthogonalisation, we may write any  $g\in \SL(2n,\CC)$ in the form
\be\label{lrdecomp}
g = k_L b_R
\ee
with unique $k_L\in K$ and $b_R\in B$. Equivalently, we  may write, with $k_R\in K$ and $b_L\in B,$
\be\label{rldecomp}
g=b_L k_R.
\ee
For present purposes, we favour the use of the $g=k_Lb_R$ decomposition and shall often drop the subscripts,
denoting the components simply as $(k,b)\in K\times B$. The natural left-multiplication action of $K$ on $\cM$
generates the ``left-handed'' action on $K\times B$ by
\be\label{leftaction}
f\underset{L}{*}(k,b) = (fk,b),\quad f\in K.
\ee
The natural right-multiplication action of $K$ on $\cM$ generates the ``right-handed'' action on $K\times B$ by
\be\label{rightaction}
f\underset{R}{*}(k,b) = (k',b') \quad\hbox{with}\quad k'b'=kbf^\dag\,.
\ee
Let us introduce the
matrix\footnote{The symbol $\1_n$ stands for the $n\times n$ identity matrix and later $id$  will stand for $\1_{2n}$.}
 $\Inn:=\diag(\1_n,-\1_n)$ and define
\be
K_+ = {\rm S}(\U(n)\times \U(n))=\{k\in K\,|\, k^\dag \Inn k=\Inn\}.
\ee

Suppose that $b\in B$ and $f\in K$. Then there exists a unique $\tilde f\in K$ such that $\tilde f b f^\dag\in B$,
and hence we get
\be\label{rightdress}
f\underset{R}{*}(k,b)  = (k\tilde f^\dag,\tilde fbf^\dag).
\ee
Moreover, this formula restricts to $K_+$, in the sense that $f\in K_+ \Leftrightarrow\tilde f\in K_+$.
The first claim is a direct consequence of the property of universal factorisation,
while the second can be checked by writing $b$ in block form, and then looking at each component separately.

The left-handed and right-handed actions naturally engender an action of $K_+ \times K_+$,
 and we shall be interested in the ring of the smooth real functions on $\cM  \simeq K\times B$
 which are invariant under this action.
 To obtain such functions,  for
${\Herm}:=\{X\in {\CC}^{2n \times 2n}
\,|\, X^\dag= X\}$ and ${\qHerm}:=\{X\in  {\CC}^{2n \times 2n}
\,|\, X^\dag=\Inn X\Inn\}$,
we introduce the maps $\Omega: \cM\rightarrow {\Herm}$ and $L: \cM\rightarrow {\qHerm}$, defined by\\
\be\label{OmLdefs}
\ba
\Omega(kb) &= bb^\dag, \\
L(kb)&=k^\dag \Inn k\Inn.
\ea
\ee
Clearly $\Omega$ and $L$ are invariant with respect to the left-handed action of $K_+$ on $\cM$.
With respect to the right-handed
action, from (\ref{rightdress}),
\be
\ba
\Omega(gf^\dag) &= \tilde f\Omega(g)\tilde f^\dag,\\
L(gf^\dag) &= \tilde fL(g)\tilde f^\dag.
\ea
\ee
From this observation there follows directly that,
with respect to the obvious conjugation  actions of $K_+$ on ${\Herm}$ and on ${\qHerm}$,
\be
\ba
&\Omega^{-1}\left(\cinf({\Herm})^{K_+}\right)\subset\cinf(\cM)^{K_+\times K_+},\\
&L^{-1}\left(\cinf({\qHerm})^{K_+}\right)\subset\cinf(\cM)^{K_+\times K_+}.
\ea
\label{invariance}\ee

Having in mind our later purpose,
we next introduce a mapping $w: \cM\rightarrow{\CC}^{2n}$ as follows. Let $\hat w\in{\CC}^{2n}$, and assume that
 $I\hat w=\hat w$; that is
 \be
 \hat w = \bpm \hat v\\0\epm,  \quad\hbox{for some fixed} \quad \hat v\in{\CC}^n.
 \label{hatw}\ee
 Define
\be\label{defw}
w(kb)=k^\dag\hat w.
\ee
From (\ref{rightdress}) we have, with respect to the right-handed action of $K_+$ on $\cM$,
\be\label{rightdressw}
w(gf^\dag) = \tilde f w(g),\qquad\forall f\in K_+,
\ee
whilst, with respect to the left-handed action of $K_+$ on $\cM$, we have the tautologous statement
\be
w(fg) = w(g),\qquad\forall f\in K_+(\hat w),
\ee
where
\be
K_+(\hat w) = \{f\in K_+\ |\ f\hat w=\hat w\}.
\label{K+w}\ee
An important relation between $L$ and $w$ ---due to the condition $I\hat w=\hat w$--- is the self-evident
\be\label{LIwrel}
L \Inn w = w.
\ee

\subsection{Poisson structure and symmetries}\label{PBandsymm}

The group decomposition $\SL(2n,{\CC})\simeq K\times B$ results in the Lie algebra decomposition
$\sl(2n,{\CC})\simeq \operatorname{Lie}(K)+\operatorname{Lie}(B)$,
and the two subalgebras $\k:=\operatorname{Lie}(K)$ and $\b:=\operatorname{Lie}(B)$ are in natural duality with one another with respect to
the invariant nondegenerate inner product on $\g:=\sl(2n,{\CC})$
\be\label{innerproduct}
\langle X,Y\rangle = \mathrm{Im}\,\tr XY,\qquad X,Y\in \g.
\ee
Consequently,  $\cM$ acquires the structure of Heisenberg double in the standard way \cite{STS}.
 That is,  $C^\infty(\cM)$ carries the (non-degenerate) Poisson bracket given by
\be\label{topPB}
\{\varphi,\psi\}(g) = \langle \nabla_g\varphi , \cR \nabla_g\psi\rangle +
\langle \nabla_g'\varphi , \cR \nabla_g'\psi\rangle,
\ee
using $ \cR \in \mathrm{End}(\g)$
provided by
half the difference of two projections, $\cR=\half(P_\k - P_\b)$, and
 $\nabla_g\varphi,\ \nabla_g'\varphi\in\g$ characterized by
\be\label{Ddefs}
\dt \varphi(e^{tX}ge^{tY}) = \langle X,\nabla_g\varphi\rangle +  \langle Y,\nabla'_g\varphi\rangle,
\qquad\forall X,Y\in\g.
\ee
With respect to this extra structure, the left-handed and right-handed actions of $K$ on $\cM$ are Poisson actions
with momentum maps $g\mapsto b_L$ and $g\mapsto b_R^{-1}$ defined by   (\ref{lrdecomp}) and (\ref{rldecomp}).

In fact, $K_+$ is a Poisson Lie subgroup of $K$ and its dual group can
be identified with $B/N$, where $N\subset B$ is the
normal subgroup of matrices having the block form,
\be\label{Ndef}
N:=\left\{\left. \bpm \1_n & X\\ 0 & \1_n\epm\right|\ X\in {\CC}^{n \times n}
\right\}.
\ee
Denoting the projection $B\rightarrow B/N$ by $\pi_N$, the momentum maps generating the left-handed and
 right-handed actions of $K_+$ on $\cM$ are respectively the maps
\be\label{Kplusmomenta}
\ba
&b_Lk_R=g\mapsto \pi_N(b_L),\\
&k_Lb_R = g\mapsto \pi_N(b_R^{-1}).
\ea
\ee

\begin{proposition}
The functions $F_l$ and $\Phi_l$, defined by
\be\label{FPhidefs}
\ba
F_l(g)&=\frac1{2l}\tr  \Omega(g)^l ,\\
\Phi_l(g)&=\frac1{2l}\tr  L(g)^l ,
\ea
\qquad l=1,2,\dots
\ee
are all invariant with respect to the action of the symmetry group $K_+\times K_+$. They form two separate
families of functions in involution on $\cM$; that is
\be
\{F_{l_1},F_{l_2}\}=0,\qquad\forall l_1,l_2,
\ee
and
\be
\{\Phi_{l_1},\Phi_{l_2}\} =0,\qquad\forall l_1,l_2.
\ee
The Hamiltonian vector field corresponding to $F_l$ is expressed in terms of the $K$ and $B$ components by
\be\label{Fflows}
{\XX}_{F_l}(g)\ :\
\left\{\ba
\dot k &= \ri k\,[\Omega^l - \nu_l\,id\,]\\
\dot b&=0
\ea\right.
\ ,
\qquad
\hbox{with}\quad\nu_l=(2n)^{-1}\tr\Omega^l.
\ee
The Hamiltonian vector field corresponding to $\Phi_l$ is expressed in terms of the $K$ and $B$ components by
\be\label{Phiflows}
\ba
{\XX}_{\Phi_l}(g)\ :\
&\left\{\ba
\dot k &=  \half \ri k(\Inn L^{l-1} - L^{l+1}\Inn - \Inn L^l + L^l\Inn)\\
\dot b &= \half \ri(id+\Inn)  L^l (id-\Inn)b.
\ea\right.\\
\ea
\ee
 Each of these vector fields generates a complete flow on $\cM$.
\end{proposition}

\begin{proof}
For both families, ($K_+\times K_+$)-invariance is obvious from (\ref{invariance}), and the involutivity properties may be deduced
directly from the forms
 of the respective Hamiltonian vector fields.
 The formula  for ${\XX}_{F_l}$ is obtained by straightforward application of the
 definitions (\ref{topPB}) and (\ref{Ddefs}).
 The derivation  for ${\XX}_{\Phi_l}$ is more lengthy, proceeding via the observation
  that $\nabla'_g\Phi_l\in\b$, which implies that $\dot g g^{-1}= -[\nabla_g\Phi_l]_\b$, and this
  can be written explicitly utilizing the fact that
 $X^\dag = -\Inn X\Inn$ entails  $X_\b := P_\b(X) = \half(id+\Inn)X(id-\Inn)$.
 The completeness property of the flow of $\XX_{F_l}$ is plain, while for $\XX_{\Phi_l}$ it follows
 by appeal to the compactness of $K$, using that $\dot{b} b^{-1}$ in (\ref{Phiflows}) depends only on $k$.
 \end{proof}

It will be important for us to have the projections of ${\XX}_{F_l}$ and ${\XX}_{\Phi_l}$ expressed in terms of $L$,
$\Omega$ and $w$. These follow directly from (\ref{Fflows}) and (\ref{Phiflows}), using (\ref{OmLdefs}) and (\ref{defw}),
and are respectively given by
\be\label{FflowsLOmw}
{\XX}_{F_l}(g)\quad\Rightarrow\quad
\left\{\ba
\dot L\Inn&=[L\Inn,\ri\Omega^l]\\
\dot\Omega&=0\\
\dot w &= -\ri[\Omega^l - \nu_l\,id\, ]w
\ea\right.
\phantom{\qquad\quad \hbox{with}\quad\nu_l=(2n)^{-1}\tr\Omega^l}
\ee
and
\be\label{PhiflowsLOmw}
X_{\Phi_l}(g)\quad\Rightarrow\quad
\left\{\ba
\dot L&=  \half \ri [ 2L^l - L^{l-1}-L^{l+1},\Inn]\\
\dot\Omega & =\half \ri(id+I)L^l(id-I)\Omega + {\half}\ri\Omega(id-I)L^l(id+I)\\
\dot w &=  \half \ri(id+\Inn)(L^l-L^{l-1})w.
\ea\right.
\ee

\subsection{ Reduction of the systems $\{F_l\}$ and $\{\Phi_l\}$}\label{reductionofFPhisystems}

In principle, one can perform reduction by setting the diagonal $n\times n$ blocks of $b_L$ and $b_R$
to arbitrary constants, elements of $B_n$, and then projecting to the quotient of
the resulting momentum constraint surface, $\cM_0$,
by the isotropy subgroup
 in $K_+\times K_+$ corresponding to the constraints.
The quotient, the reduced phase space $\cM_\red$, is naturally a smooth symplectic manifold
if standard regularity conditions are met (see e.g.~\cite{Lu}).
The functions $F_l$ and $\Phi_l$ then descend to smooth functions $F_l^\red$ and $\Phi_l^\red$
on $\cM_\red$ forming Abelian Poisson algebras with respect to the reduced symplectic structure.
The isotropy group of the constraints is also known as the \textit{gauge group},
and  the associated transformations of $\cM_0$ are often called \textit{gauge transformations}.

The following result gives us a device (used already in \cite{M,M1})
whereby the momentum constraints are expressed as explicit functions of $g\in\cM$.
The proof is a simple exercise.

\medskip
\begin{proposition}\label{momentummap}
Suppose $\mu_1, \mu_2, \tilde \mu_1, \tilde\mu_2 \in B_n$ are given.  The condition
\be
\ba
\cM\owns g = k_Lb_R \quad\hbox{with}\quad
b_R = \bpm \mu_1 & * \\ 0 & \mu_2\epm
\ea
\ee
is equivalent to
\be
g^\dag g - g^\dag g\bpm (\mu_1^\dag \mu_1)^{-1}&0\\0&0\epm g^\dag g = \bpm 0&0\\0& \mu_2^\dag \mu_2\epm,
\ee
 and the condition
\be
\ba
\cM\owns g = b_Lk_R \quad
\hbox{with}\quad
b_L = \bpm \tilde\mu_1 & * \\ 0 & \tilde\mu_2\epm
\ea
\ee
is equivalent to
\be
gg^\dag - gg^\dag \bpm 0 &0\\0& (\tilde\mu_2\tilde \mu_2^\dag)^{-1}\epm gg^\dag  =
\bpm \tilde\mu_1\tilde \mu_1^\dag &0\\0& 0 \epm.
\ee
\end{proposition}

\medskip
In the present  work, we study reduction under the following constraints.
We choose real, positive numbers $x,y,\alpha$, supposing additionally that $\alpha<1$, and then
fix  the constraint surface $\cM_0$ by
\be\label{constraint}
\cM_0 := \left\{g\in\cM\ \left| \ b_R = \bpm x\1_n &*\\ 0& x^{-1}\1_n\epm,\
b_L = \bpm y^{-1}\sigma &*\\ 0& y\1_n \epm\right\}\right. ,
\ee
where $\sigma$ is an element of $B_n$, defined in relation to the previously
chosen vector $\hat v$ in (\ref{hatw}) by the property that
\be\label{sigmadef}
\sigma\sigma^\dag = \alpha^2 \1_n + \hat v\hat v^\dag.
\ee
This presupposes the condition on the fixed vector $\hat v$ that
$|\hat v|^2=\alpha^2(\alpha^{-2n} -1)$,
thus ensuring  that $\det(\sigma)=1$.
The right-hand part of the corresponding isotropy group is  the whole of $K_+$.
The left-hand part of the isotropy group, denoted $K_+(\sigma)$  (since it depends only
on the choice of the element $\sigma$), is
the direct product
\be
K_+(\sigma)= K_+(\hat w) \times {\TT}_1,
\ee
with $K_+(\hat w)$ in (\ref{K+w}) and with ${\TT}_1$ given by
\be
{\TT}_1 := \{\hat \gamma:= \diag(\gamma \1_n, \gamma^{-1} \1_n)\mid \gamma\in \U(1)\}.
\label{T1}\ee
Here, the ${\TT}_1$ factor of $K_+(\sigma)$ acts
on the vector $w$ (\ref{defw}) according to the rule
\be
\hat \gamma: w \mapsto \gamma^{-1} w.
\label{gamma}\ee
The task is to characterize the quotient,
\be
\cM_\red := \cM_0/(K_+(\sigma) \times K_+).
\label{Mred}\ee

The approach followed in \cite{FG} mimics that of \cite{M,FP}, and results in a model of $\cM_\red$
(proved in \cite{FG} to be  a smooth manifold)
for which the functions $F_l^\red$ are presented as a collection of
interesting commuting Hamiltonians, and the
$\Phi_l^\red$ are trivial.
It proceeds, after imposing the constraints, by using the isotropy subgroups
for both the left-handed and right-handed actions
to bring  $k$
to the form
\be  k=
\bpm \varrho&0\\0& \1_n\epm
\bpm \cos(q)&\ri\sin(q)\\ \ri\sin(q)&\cos(q) \epm
\quad\hbox{with}\quad q =\diag(q_1,\dots, q_n),\quad \varrho\in \mathrm{SU}(n).
\label{bLform}\ee
In essence, the result develops from finding the
explicit dependence of  the matrix $\Omega$ as a function of $L$, i.e.~of $q$,  and of conjugate variables,
such that the constraint is satisfied.

Alternatively, in the current article
we shall look for a model of the reduced phase space for which the functions $\Phi_l^\red$
form a set of interesting commuting
Hamiltonians and the $F_l^\red$ are trivial. This is achieved by using the right-hand isotropy subgroup
to bring $\Omega$ to blockwise diagonal form, following which the reduction proceeds by
representation, via constraints,
of the matrix $L$ as a function of $\Omega$ and of canonically conjugate variables.
Our objective in the next section is to elaborate this statement in detail.

\section{Analysis of the reduced system}

We start with the observation that, for any $g= kb$ from the constraint surface $\cM_0$,
the right-handed action of $K_+$ may be used to bring $b$ to the form
\be\label{bRform}
b= \bpm x \1_n&\beta\\0&x^{-1}\1_n\epm \quad \hbox{with } \beta =
\diag(\beta_1,\dots,\beta_n),\quad \beta_i\in{\RR},\,\,\, \beta_1\geq\dots\geq\beta_n\geq0.
\ee
This is an application of the standard singular value decomposition of $n\times n$ complex matrices.
The $\beta_i$ are invariants on $\cM_0$ with respect to the full gauge group $K_+(\sigma) \times K_+$.
Now the idea is to introduce a partial gauge fixing where $b$ has the above form,
and label  the points of $\cM_\red$ (\ref{Mred}) by
the $\beta_i$ together with  further
invariants with respect to the residual gauge transformations.
In what follows we \textit{assume}  that
\be
\beta_1 > \beta_2 > \dots > \beta_n >0.
\label{An2}\ee
Then the residual gauge group is $K_+(\sigma) \times {\TT}_{n-1}$, where ${\TT}_{n-1}$ contains
the matrices of the form
$\diag(\tau, \tau)$, with $\tau = \diag(\tau_1,\dots, \tau_n)$ and $\tau_k\in \U(1)$ subject to the
condition $\prod_{k=1}^n \tau_k^2 =1$.
It is readily seen that,
 with $w=w(g)$ defined in (\ref{defw}), the triple $(\beta, w, L)$ provides a complete set of invariants with
respect to the factor $K_+(\hat w)$ of the residual gauge group.
After factoring this out,  we combine the residual right-handed  gauge group ${\TT}_{n-1}$ and
the factor ${\TT}_1$ (\ref{T1}) of $K_+(\sigma)$,
which acts by (\ref{gamma}),  into the $n$-torus
\be
{\TT}_n = \{ T = \diag(\tau, \tau) \mid \tau=\diag(\tau_1,\dots, \tau_n),\,\,\, \tau_i\in \U(1)\}.
\label{Tn}\ee
The residual gauge transformation by $T\in {\TT}_n$
 acts on the triple $(\beta, w, L)$ according to
\be
T: (\beta, w, L) \mapsto (\beta, T w,  T L T^\dagger).
\label{Tact}\ee

In the next subsection,  we solve the constraint condition and express
$w$ and $L$, up to the gauge action (\ref{Tact}), in terms of $\beta$ and
further invariants.
In Subsections \ref{PBandsymm} and \ref{reductionofFPhisystems} we construct
Darboux coordinates on $\cM_\red$
and determine the form of the reduced Hamiltonian $\Phi_1^\red$ in terms of them.

The assumption (\ref{An2})
can certainly be made by restriction
to an open subset of $\cM_0$.
We shall adopt further similar  assumptions in our arguments below;
requiring  various functions to be non-vanishing before we divide by them.
 As will be explained in \cite{LONG}, it can be proved that our  analysis covers a \textit{dense}
  open subset of $\cM_\red$.
The domain  on which our subsequently derived local formulae are valid is  revisited in
 Section \ref{discussion}.

\subsection{Solving the constraint conditions}

So far we have introduced partial gauge fixing so that $b=b_R$ takes the form
 specified in (\ref{bRform}), and then adopted (\ref{An2}).
Now we deal with the consequences of the left-hand part of the constraints imposed in (\ref{constraint}).
 According to Proposition \ref{momentummap}, this is equivalent to
\be\label{firststepconstraint}
gg^\dag - gg^\dag \bpm 0 &0\\0& y^{-2}\epm gg^\dag  = \bpm y^{-2}\sigma\sigma^\dag &0\\0& 0 \epm.
\ee
Substituting $g=kb$, then conjugating with $k^\dag$ and multiplying by $2y^2$, we have
\be
2y^2 bb^\dag -  bb^\dag k^\dag (id - \Inn) k bb^\dag = 2 k^\dag\bpm \sigma\sigma^\dag&0\\0&0\epm k
\ee
and, after using (\ref{sigmadef}) and rewriting the matrix on the right hand side accordingly, we obtain
\be\label{mainconstraint}
2y^2\Omega - \Omega^2 + \Omega L\Inn\Omega = \alpha^2id + \alpha^2LI + 2ww^\dag.
\ee
Our objective is to find the general solution of (\ref{mainconstraint}) for $L$ in terms of
\be
\Omega = b b^\dagger =\bpm x^2\1_n + \beta^2 & x^{-1}\beta \\ x^{-1}\beta & x^{-2}\1_n\epm.
\label{Ombeta}\ee

Somewhat surprisingly, containing as it does the  several
unknown quantities, $ww^\dag$ and $L$,  equation (\ref{mainconstraint}) can be solved directly.
 To see this, we start by noticing that the simple blockwise diagonal form of $\Omega$
 allows us to diagonalise it very easily.
To present $\Omega$ in diagonalised form,
let us introduce  the matrix
\be\label{rhoform}
\rho := \bpm \G&\S\\ \S &-\G\epm\quad\hbox{with}\quad
\G:=\diag(\Gamma_1,\dots,\Gamma_n),\quad \S:=\diag(\S_1,\dots,\S_n).
\ee
Define $\Gamma_i$ and $\S_i$ by the formulae
\be\label{rhoentries}
\Gamma_i = \left[\frac{\Lm_i-x^{-2}}{\Lm_i-\Lm_i^{-1}}\right]^{\frac12}\ ,
\qquad
\S_i = \left[\frac{x^{-2} - \Lm_i^{-1}}{\Lm_i-\Lm_i^{-1}}\right]^{\frac12}
\ee
in terms of the new variables
\be
\Lm_1 > \Lm_2 > \dots > \Lm_n > \max(x^2,x^{-2}).
\label{Lambdaord}\ee
Then it is readily checked that every matrix $\Omega$ (\ref{Ombeta}) can be written in form
\be
\Omega = \rho\, \diag(\Lambda_1,\dots, \Lambda_n, \Lambda_{n+1},\dots, \Lambda_{2n})\, \rho
\quad\hbox{with}\quad \Lambda_{n+i} = \Lambda_i^{-1},
\label{Omdiag}\ee
using the following invertible correspondence between the variables $\beta_i$ and $\Lambda_i$:
\be\label{betaentries}
\beta_i = \left[\Lm_i+\Lm_i^{-1} - x^2 - x^{-2}\right]^\frac12 .
\ee
Because of the blockwise diagonal structure of $\Omega$, it is enough to check the
claim for the case $n=1$.
The condition (\ref{Lambdaord}) is equivalent to (\ref{An2}).
The relations $\Gamma_i^2 + \Sigma_i^2=1$ entail that $\rho$ is a symmetric real orthogonal matrix,
\be
\rho =\bar \rho =  \rho^\dagger = \rho^{-1}.
\label{rhoprop}\ee

Now we return to (\ref{mainconstraint}), from now on using the variables
$\Lambda_i$ instead of the  variables $\beta_i$.
Setting
\be\label{Qandw}
Q := \rho L\Inn\rho
\quad\hbox{and}\quad
\tw := \rho  w,
\ee
we get
\be
2y^2\Lm - \Lm^2 + \Lm Q \Lm = \alpha^2id + \alpha^2Q + 2\tw\tw^\dag.
\ee
Assuming that we can divide, this gives in components
\be\label{secondconstraintversion}
Q_{ab} = (\Lm_a\Lm_b - \alpha^2)^{-1}\Bigl[(\Lm_a^2 - 2y^2\Lm_a
+ \alpha^2)\delta_{ab} + 2\tw_a {\tw_b}^*\Bigr],
\quad a,b=1,2,\dots,2n.
\ee
 Reformulating   (\ref{LIwrel}), we have
\be
 Q\tw = \tw;
\ee
 that is
\be
\tw_a=(Q\tw)_a = \sum_{b=1}^{2n} Q_{ab}\tw_b =
\frac{(\Lm_a^2+\alpha^2-2y^2\Lm_a)}{(\Lm_a^2-\alpha^2)}\tw_a +
2\tw_a\sum_{b=1}^{2n}\frac{|\tw_b|^2}{\Lm_a\Lm_b-\alpha^2}.
\ee
Supposing that $\tw_a\neq0$, this yields
\be\label{modwrel}
\sum_{b=1}^{2n}\frac{|\tw_b|^2}{\Lm_a\Lm_b-\alpha^2}
=
\frac{y^2\Lm_a - \alpha^2}{\Lm_a^2-\alpha^2},
\ee
from which each of the $|\tw_b|^2$ is expressed in terms of the
components of $\Lm$, by means of the inverse of the Cauchy--like
matrix $C_{ab} = (\Lm_a\Lm_b-\alpha^2)^{-1}$.

Working on the open domain where (\ref{An2}) and all non-vanishing assumptions
 hold, we find explicit expressions for  $|\tw_a|^2$ as functions of $\Lambda$.

\begin{proposition}
Solving (\ref{modwrel}), we obtain
\be\label{modtwsol}
|\tw_a|^2 = \alpha(\Lm_a-y^2)\prod_{\substack{b= 1\\ (b\neq a)}}^{2n}
\frac{\alpha^{-1}\Lm_a\Lm_b-\alpha}{\Lm_a-\Lm_b},
\qquad a=1,\dots,2n.
\ee
\end{proposition}
\begin{proof}
Rewriting (\ref{modwrel}), we have
\be
|\tw_a|^2 = \sum_{b=1}^{2n}(C^{-1})_{ab}\frac{\alpha^{-1}y^2x_b - 1}{x_b^2-1},
\ee
with
\be\label{xdefs}
C_{ab}=\frac{\alpha^{-2}}{x_a x_b-1},\qquad x_a=\alpha^{-1}\Lm_a.
\ee
From the standard formula for the inverse of a Cauchy matrix, we may deduce
\be
(C^{-1})_{ab} = \alpha^2\frac{(x_ax_b)^{2n}}{(x_ax_b-1)}
\frac{A(x_a^{-1})A(x_b^{-1})}{A'(x_a)A'(x_b)},\qquad a,b=1,2,\dots,2n,
\ee
using the complex function
\be
A(z) = \prod_{a=1}^{2n}(z-x_a)
\ee
and its derivative  $A'(z)$.
Consequently,
\be\label{explicitly}
|\tw_a|^2 = \frac{\alpha^2x_a^{2n}A(x_a^{-1})}{A'(x_a)}
\sum_{b=1}^{2n}\frac{x_b^{2n}A(x_b^{-1})}{(x_ax_b-1)A'(x_b)}\frac{\alpha^{-1}y^2x_b-1}{x_b^2-1}.
\ee
To simplify the sum, introduce the rational function $\Psi_a(z)$ of a complex variable
\be\label{Phia}
\Psi_a(z) := \frac{z^{2n}A(z^{-1})(\alpha^{-1}y^2z- 1)}{(x_az-1)(z^2-1)A(z)}.
\ee
Observing that $\Psi_a(z)dz$ extends to a meromorphic 1-form on the Riemann sphere $\overline{\CC}$, the sum of its
residues over $\overline{\CC}$ must be zero. All the poles of $\Psi_a(z)dz$ are simple, and
they are located at $z=x_b$ for $b=1,2,\dots,2n$ and at $z=\pm1$.
The sum of the
residues at $z=x_b$ is exactly the sum in (\ref{explicitly}), and so this sum
can be evaluated by computing the residues at $z=\pm1$.
 We find
\be
 \rez_{z=+1}\Bigl( \Psi_a(z)dz \Bigr)+  \rez_{z=-1}\Bigl( \Psi_a(z)dz\Bigr)
=
-\frac{x_a-\alpha^{-1}y^2}{x_a^2-1}.
\ee
Substitution into (\ref{explicitly}) produces
\be
|\tw_a|^2 = \alpha\left(\frac{\alpha x_a-y^2}{x_a^2-1}\right) \frac{x_a^{2n}A(x_a^{-1})}{A'(x_a)},
\ee
and replacing $x_a=\alpha^{-1}\Lm_a$ gives the stated result.
\end{proof}

We have expressed $Q$ (\ref{Qandw}), and therefore also $L=\rho Q \rho I$, in
terms of $\Lambda$ and $\tw=\rho w$.
Hence it follows from (\ref{betaentries})  and the transformation rule (\ref{Tact})
 that we may parametrize the gauge orbits using $\Lambda$
 together with invariants of $w$.
Equivalently,
 we may build invariants out of $\tilde w$, which, due to the form of $\rho$
(\ref{rhoform}), transforms under the residual
gauge action (\ref{Tact}) in the same way as $w$, i.e.,
\be
T:  \tilde w \mapsto T \tilde w.
\ee
Recalling the form of $T\in {\TT}_n$ (\ref{Tn}), we see  that the angles
$\theta_j$ defined by the relations
\be
\tilde w_j^* \tilde w_{n+j} = \vert \tilde w_j \tilde w_{n+j}\vert e^{\ri \theta_j},  \qquad j=1,\dots, n,
\label{thetadef}\ee
are invariants.
Since the conditions $\tilde w_j \in {\RR}_{>0}$ for all
 $j=1,\dots, n$ define a complete gauge fixing for the residual gauge transformations (3.4),
the variables $\Lambda_j$ together with the $\theta_j$ provide a \textit{complete set of invariants}
that label the gauge orbits in our open subset of $\cM_0$.

\subsection{Darboux coordinates on the reduced space}

The reduced phase space $\cM_\red$ is a symplectic manifold,  and we denote the
 Poisson bracket of smooth functions
 on $\cM_\red$ by $\{\ ,\ \}_\red$.
It is apparent already in (\ref{FflowsLOmw}) that the eigenvalues of $\Omega$ and
 the phase-like invariants of $\tw$, as exhibited in (\ref{thetadef}), are candidates
 for Darboux coordinates.  We are going to  prove that they indeed are such.
As a preparation, we next formulate a consequence of the general theory
of Hamiltonian reduction.

Let $\cM_1$ denote the subspace of the constraint surface $\cM_0$ (\ref{constraint})
consisting of the elements for which $b$ has the form (\ref{bRform}).
 Then there is a natural one-to-one correspondence between the
 gauge invariant smooth functions on $\cM_1$, with respect to the
 residual gauge transformations acting on $\cM_1$,
and the smooth functions on $\cM_\red$ (\ref{Mred}).
Take a $(K_+\times K_+)$-invariant function $\cH$ on $\cM$ and a gauge invariant function
$\cG$ on $\cM_1$, and consider the Poisson bracket $\{\cG^\red, \cH^\red\}_\red$ of
the corresponding functions $\cG^\red$ and $\cH^\red$ on $\cM_\red$.
The gauge invariant function on $\cM_1$ that corresponds to $\{ \cG^\red, \cH^\red\}_\red$
is the derivative of $\cG$ along any vector field of the form
\be
\XX_\cH^1 = \XX_\cH + \YY_\cH,
\label{R14}\ee
where $\XX_\cH$ is the Hamiltonian vector field of $\cH$ restricted to $\cM_1$, and $\YY_\cH$
represents the right-handed action
of point dependent elements of the Lie algebra $\k_+$ of $K_+$,  chosen in such a way
that $\XX_\cH^1$ is tangent to $\cM_1$.
This is expressed by the equality
\be
\{ \cG^\red, \cH^\red\}_\red = (\XX_\cH^1(\cG))^\red.
\label{R15}\ee
The vector field $\XX_\cH^1$ is determined in the following way.
If $\dot{k}$ and $\dot{b}$ denote the components of $\XX_\cH(g)$ corresponding to
the decomposition $\cM\owns g=kb$,  and $k'$ and $b'$ denote the components of $\XX^1_\cH(g)$
corresponding to the decomposition $\cM_1\owns g=kb$, then we have
\be\label{R16}
k' = \dot{k} - kY, \quad b' = \dot{b} + [Y, b],
\ee
where $Y \in \k_+$  is the ``compensating infinitesimal gauge transformation'', ensuring that the
$\XX_\cH^1$-derivative $b'$ of $b$ is
consistent with the form of $b$ (\ref{bRform}).
This fixes  $Y$ up to  infinitesimal, right-handed gauge transformations tangent to $\cM_1$.
Concretely,  writing $Y=\diag(Y_1,Y_2)$, the $B$-component of (\ref{R16})
can be recast as
\be
\beta'= \dot{b}_{12}+ Y_1\beta - \beta Y_2,
\label{betaprime}\ee
where $\dot{b}_{12}$  denotes the top-right $n\times n$ block of $\dot{b}$.
The condition on $Y$ is that $\beta'$ must be a real diagonal matrix, because $\beta$
is a real diagonal matrix.
We observe from (\ref{betaprime}) that, up to its inherent ambiguity, $Y$
can be viewed as a function of $\beta$ and $\dot{b}_{12}$, which
themselves are functions on $\cM_1$.

We shall apply the above  procedure to the open submanifold $\check \cM_\red$ of $\cM_\red$
that can  be parametrized
by the invariants $\Lambda_j$ (\ref{Omdiag})  and $e^{\ri \theta_j}$ (\ref{thetadef}),
and denote the corresponding submanifold of $\cM_1$ by $\check \cM_1$
We note that every gauge invariant function on $\check \cM_1$ can be regarded as a function of $\beta$ and $w$,
since they determine $L$ by equations  (\ref{betaentries})-(\ref{secondconstraintversion}).
For a gauge invariant function $\cG$ on $\check \cM_1$, denoting by $\cG^\red$ the expression in the local coordinates
$(\Lm,e^{i\theta})$ of the corresponding function on $\check \cM_\red$, we have
\be
\cG^\red(\Lm,e^{i\theta}) = \cG(\beta,w),
\ee
where $(\beta,w)\mapsto(\Lm,e^{i\theta})$ is given by (\ref{betaentries}), (\ref{Qandw}) and (\ref{thetadef}).
We shall also use the fact that on
$\check \cM_1$ the functions $\vert\tilde w_a\vert$ ($a=1,\dots, 2n$) are non-zero and depend only on $\Lambda$.

\begin{theorem}
On the open submanifold of $\check \cM_\red \subset \cM_\red$ parametrized by $\lambda_j:= \frac{1}{2} \log \Lambda_j$
(\ref{Omdiag})
and the angles $\theta_j$ (\ref{thetadef}) we have the canonical Poisson brackets
\be
\{\lambda_j, \lambda_l\}_\red =0, \quad \{\theta_j, \lambda_l\}_\red = \delta_{jl},\quad
\{ \theta_j, \theta_l\}_\red = 0,\quad j,l=1,\dots, n.
\label{Darboux}\ee
\end{theorem}

\begin{proof}
The first two relations in (\ref{Darboux}) are shown easily.
For this, we start by pointing out that
the reductions
of the Poisson commuting functions $F_l\in C^\infty(\cM)^{K_+\times K_+}$,
defined in (\ref{FPhidefs}), read
\be
F_l^\red = \frac{1}{l} \sum_{j=1}^n \cosh(2 l\lambda_j).
\label{Flred}\ee
The identity $\{ F_j^\red, F_l^\red\}_\red =0$  for all $j,l$  is
assured by the reduction, and
is clearly equivalent to $\{ \lambda_j, \lambda_l\}_\red =0$.

Direct calculation on the reduced phase space  gives
\be
\{e^{\ri \theta_j}, F_l^\red\}_\red = 2 \ri e^{\ri \theta_j}
\sum_{m=1}^n \{ \theta_j, \lambda_m\}_\red \sinh(2l \lambda_m).
\label{R7}\ee
Notice from (\ref{Fflows}) that the Hamiltonian vector field of $F_l$ is tangent to $\cM_1$.
Calculating  the right-hand-side of (\ref{R15})
for $\cH= F_l$ and
$\cG = e^{\ri \theta_j}$ defined by (\ref{thetadef}), we find from (\ref{FflowsLOmw}) that
\be
\XX^1_{F_l}(e^{\ri \theta_j}) = 2 \ri\, e^{\ri \theta_j} \sinh(2 l \lambda_j).
\ee
Equality between the last two expressions is equivalent to  $\{ \theta_j, \lambda_l\}_\red = \delta_{jl}$.

The Jacobi identity for $\{\ ,\ \}_\red$ and
the formulae
$\{ \theta_i, \lambda_j\}_\red = \delta_{ij}$ imply
that the functions
\be
 P_{kl}:=\{ \theta_k, \theta_l\}_\red
\label{R13}\ee
depend only on $\lambda$.
It remains  to prove that these functions vanish identically.

We consider the
function $\Phi_1\in C^{\infty}(\cM)^{K_+ \times K_+}$, also defined in (\ref{FPhidefs}).
The Hamiltonian vector field of $\Phi_1$, given by the $l=1$ special case of (\ref{Phiflows}),
is tangent to $\cM_0$, but is not tangent to $\check \cM_1$. In this case $\dot{b}_{12}= 2 \ri x^{-1} L_{12}$,
and we can find $Y=Y(\beta,2\ri x^{-1}L_{12})\in\k_+$ such that
\be
\beta' \equiv \XX_{\Phi_1}^1(\beta) = 2 \ri x^{-1} L_{12}  + Y_1 \beta - \beta Y_2
\label{R24}\ee
will be a real diagonal matrix.
To proceed further, we point out that for every element $g = k b \in \check \cM_1$, there exists
another element $g^\sharp = k^\sharp b \in \check \cM_1$ for which
\be\label{wgsharp}
w(g^\sharp) = w(g)^*
\quad\hbox{and consequently}\quad
L(g^\sharp) = L(g)^*,
\ee
where star denotes complex conjugation. This holds since the constraint  condition
(\ref{mainconstraint})  is stable under
complex conjugation\footnote{We can take
$g^\sharp = g^*$  whenever the fixed vector $\hat w$ (\ref{hatw}) is real.}.
More concretely, it reflects the fact
that for fixed $\beta$  the constraints determine
only the moduli $|\tw_a|$ of the $\tw_a$ (\ref{Qandw}), and all values are possible for $\arg(\tw_a)$.
For a given $g$, any two choices of $g^\sharp$ are related by a gauge transformation, since  $w$  determines
$k$ up to the left-handed action of
$K_+(\hat w)$.
The rest of the proof relies on the property
\be\label{YLcondition}
Y(\beta, 2\ri x^{-1}L_{12}^*) = Y(\beta, 2\ri x^{-1}L_{12})^T,
\ee
which follows by comparison of equation (\ref{R24}) with its complex conjugate.
Of course, this equality is understood up to the ambiguity in  $Y$,  that does not  affect the derivatives
of gauge invariant functions.

Let $A=\diag(A_1,A_2,\dots,A_n)$ be a diagonal matrix with $A_j\in\RR$ for all $j$, and introduce the $2n\times 2n$ matrix
\be
\hat A = \bpm 0&-A\\ A&0\epm.
\ee
We then define the gauge invariant function $\cG_A$  on $\check \cM_1$ by
\be
 \cG_A(g) =  \frac1{2\ri}w^\dag\hat Aw.
\ee
Using the $l=1$ case of $\dot{w}$ from  (\ref{PhiflowsLOmw}),  with (\ref{R16}) and (\ref{LIwrel}),
the derivative $w'$ of $w$ along $\XX_{\Phi_1}^1$  reads
\be\label{wcompphi1}
w' = \half\ri (id+I)L(id-I)w + Yw,
\ee
and we easily check that
\be
\XX_{\Phi_1}^1(\cG_A)(g)=\XX_{\Phi_1}^1(\cG_A)(g^\sharp).
\label{ourfavouritetrick}\ee
Indeed, denoting $Y(\beta,2\ri x^{-1}L_{12})$ simply by $Y$ for short, we have
\be
\ba
\XX_{\Phi_1}^1&( \cG_A)(g)
 =
w'^\dag\hat A w + w^\dag\hat A w'  \\
{}\quad &
 =
\frac1{2\ri}w^\dag\Bigl( \left[\half\ri (id-I)L^\dag(id+I) + Y\right]\hat A^T  +
\hat A\left[\half \ri (id+I)L(id-I) + Y\right] \Bigr)w\\
\ea
\ee
and, using(\ref{wgsharp}) and (\ref{YLcondition}),
\be
\ba
{\XX}_{\Phi_1}^1( \cG_A)(g^\sharp)
=
\frac1{2\ri}w^T\Bigl(\hat A\left[\half \ri (id+I)L^*(id-I) + Y^T\right] + \left[\half\ri (id-I)L^T(id+I) + Y^T\right]\hat A^T \Bigr)w^*.
\ea
\ee
 It is easy to see that these are the same.

 Next, let us inspect the reduced version of the equality (\ref{ourfavouritetrick}).
Taking into account the relation $\tw=\rho w$ and using $\rho\hat A\rho=-\hat A$, we obtain
\be
\cG_A^\red(\lambda,\theta) = \sum_{i=1}^nA_i \bigl(|\tw_i|\,|\tw_{n+i}|\bigr)(\lambda)\sin\theta_i.
\ee
On the other hand, $\Phi_1^\red$ takes the form
\be
\Phi_1^\red(\lambda, \theta)= V(\lambda) + \sum_{j=1}^n f_j(\lambda) \cos \theta_j
\label{R19}\ee
with some  functions $V$ and $f_j$.
(Equation (\ref{distilham}) below shows that $f_j(\lambda) \neq 0$ on  $\check\cM_\red$.)
 Direct calculation then yields
\be
\ba
\{\cG_A^\red,\Phi_1^\red\}_\red  &=
\sum_{i=1}^n\left[\frac{\partial \cG_A^\red}{\partial\theta_i}\frac{\partial\Phi_1^\red}{\partial\lambda_i}-\frac{\partial \cG_A^\red}{\partial\lambda_i}
\frac{\partial\Phi_1^\red}{\partial\theta_i}\right]
+\sum_{i,j=1}^nP_{ij}\frac{\partial\cG_A^\red}{\partial\theta_i}\frac{\partial\Phi_1^\red}{\partial\theta_j}\\
&= \sum_{i=1}^n\sum_{j=1}^n A_jf_i\frac{\partial\left(|\tw_j|\,|\tw_{j+1}|\right)}{\partial\lambda_i}\sin\theta_i\sin\theta_j\\
& +
\sum_{i=1}^n A_i|\tw_i|\,|\tw_{n+i}|\cos\theta_i
\left[\sum_{j=1}^n\frac{\partial f_j}{\partial\lambda_i}\cos\theta_j + \frac{\partial V}{\partial\lambda_i} -
\sum_{j=1}^nf_j P_{ij}\sin\theta_j\right],\\
\ea
\ee
with the notation (\ref{R13}).
This implies the relation
\be
\{\cG_A^\red,\Phi_1^\red\}_\red(\lambda,-\theta) - \{\cG_A^\red,\Phi_1^\red\}_\red(\lambda,\theta) =
2 \sum_{i=1}^n \sum_{j=1}^n A_i \cos\theta_i \Bigl(|\tw_i|\,|\tw_{n+i}| P_{ij} f_j \Bigr)(\lambda) \sin\theta_j.
\label{critexpr}\ee
Now we  notice from (\ref{thetadef})  that, for invariant functions on $\check \cM_1$,
$(\lambda,\theta)\mapsto(\lambda,-\theta)$ is equivalent to $\tw\mapsto\tw^*$ and, as $w=\rho\tw$,
the same is true for $w$, i.e. $w\mapsto w^*$.
Therefore, taking into account also (\ref{R15}) and (\ref{wgsharp}), the reduced version of the equality (\ref{ourfavouritetrick})
says that the expression in (\ref{critexpr})  is zero.
Choosing
\be
A_i = \delta_{ik},\quad \theta_j = -\frac\pi2\delta_{jl}
\ee
we  obtain
\be
2 |\tw_k|\,|\tw_{n+k}| f_l P_{kl}=0.
\ee
This necessitates the vanishing of $P_{kl}$, whence the proof is complete.
\end{proof}

\subsection{The  form of the Hamiltonian $\Phi_1^\red$}

The Hamiltonian of interest is  the reduction of $\Phi_1$---the simplest element in the ring
of invariant functions of $L$---expressed as a function of the Darboux coordinates
$\lambda_j$, $\theta_j$ (\ref{Darboux}) on
the reduced phase space. The desired expression can be derived by evaluation of the formula
\be
\Phi_1^\red (\lambda,\theta)  \simeq {\half}\tr L |_{\check\cM_1}
\ee
using, on account of (\ref{Qandw}), $L= \rho Q \rho I$  with $Q$ given by (\ref{secondconstraintversion}).
Since $\tr L$ is gauge invariant, we obtain $\Phi_1^\red$ as a function of $\lambda, \theta$ if we
substitute (\ref{modtwsol}) and (\ref{thetadef}).
In agreement with \cite{FG}, let us replace
\be\label{xyalphsubs}
\alpha=e^{-\mu},\quad x=e^{-v},\quad y=e^{-u},
\ee
where $u,v,\mu$ are real parameters, $\mu>0$.
We shall prove the following
\begin{theorem}\label{hamthm}
The reduced Hamiltonian $\Phi_1^\red$ takes the form
\be\label{distilham}
\ba
 \Phi_1^\red (\lambda,\theta)&=
V(\lambda) +  e^{v-u}\sum_{k=1}^n\frac{\cos\theta_k}{\cosh^2\lambda_k}
\left[1 - \frac{\sinh^2v}{\sinh^2\lambda_k}\right]^{1/2} \left[1 - \frac{\sinh^2u}{\sinh^2\lambda_k} \right]^{1/2}\\
&\qquad\times
\prod_{\substack{l=1\\(l\neq k)}}^n \left[1 - \frac{\sinh^2\mu}{\sinh^2(\lambda_k - \lambda_l)}\right]^{1/2}
\left[1 - \frac{\sinh^2\mu}{\sinh^2(\lambda_k + \lambda_l)}\right]^{1/2}
\ea
\ee
with
\be
V(\lambda) =e^{v-u}\left(\frac{\sinh(v)\sinh(u)}{ \sinh^2\mu}
\prod_{k=1}^n\left[1 - \frac{\sinh^2\mu}{\sinh^2\lambda_k} \right]
-\frac{\cosh( v)\cosh(u)}{\sinh^2\mu}
\prod_{k=1}^n\left[1 + \frac{\sinh^2\mu}{\cosh^2\lambda_k} \right]
+C\right)
\label{pot}\ee
where $\displaystyle{C= ne^{u-v} + \frac{\cosh( v -u)}{\sinh^2\mu}}$.
\end{theorem}

\medskip
\begin{proof}
Let us write
\be
Q= D + 2\cW C\cW^\dag
\ee
where, from (\ref{secondconstraintversion}),
\be\label{DWC}
\ba
D_{ab}&=\delta_{ab}D_a \quad\hbox{with}\quad
D_a = (\Lm_a^2-\alpha^2)^{-1}(\Lm_a^2+\alpha^2-2y^2\Lm_a),\\
\cW_{ab} &= \tw_a\delta_{ab}, \quad\hbox{and}\quad C_{ab} = (\Lm_a\Lm_b-\alpha^2)^{-1}.
\ea
\ee
Hence, using (\ref{Qandw}) together with (\ref{rhoform}), we have
\be\label{rawham0}
\ba
\Phi_1^\red &= \frac12\tr Q\rho I\rho
=
\frac12\tr (D + 2\cW C\cW^\dag)\bpm \G^2-\S^2 & 2\G\S\\ 2\G\S & -\G^2+\S^2\epm\\
&=
\frac12\sum_{k=1}^n(\G_k^2-\S_k^2)\Bigl[D_k - D_{n+k} + 2C_{kk}|\tw_k|^2 - 2C_{n+k,n+k}|\tw_{n+k}|^2\Bigr]\\
&\qquad+
2\sum_{k=1}^n\G_k\S_kC_{k,n+k}(\tw_k\tw_{n+k}^*+\tw_k^*\tw_{n+k}).
\ea
\ee
Substituting from (\ref{DWC}), (\ref{rhoentries}) and then reorganising terms, we get
\be\label{rawham}
\ba
\Phi_1^\red
&=
\frac12\sum_{a=1}^{2n}
\frac{\Lm_a+\Lm_a^{-1}-2x^{-2}}{\Lm_a-\Lm_a^{-1}}
\left( \frac{\Lm_a^2+\alpha^2-2y^2\Lm_a}{\Lm_a^2-\alpha^2}
+ \frac{2|\tw_a|^2}{\Lm_a^2-\alpha^2}\right)\\
&\qquad+
4\sum_{k=1}^n\left[\frac{(\Lm_k-x^{-2})(x^{-2}-\Lm_k^{-1})}{(\Lm_k-\Lm_k^{-1})^2}\right]^{\frac12}
\frac{|\tw_k|\,|\tw_{n+k}|\cos\theta_k}{1-\alpha^2}.
\ea
\ee

Let us denote by $V$ the first sum in formula (\ref{rawham}), and insert  $|\tw_a|^2$ from (\ref{modtwsol}).
Introducing the complex function $\Psi(z)$ by
\be
\Psi(z) = F(z) + G(z)
\ee
with
\be
\ba
F(z) &=
\alpha^2\frac{(z^2-2x^{-2}z+1)(z-y^2)}{(z^2-1)(z^2-\alpha^2)^2}
\prod_{a=1}^{2n}\frac{(\alpha^{-1}z\Lm_a-\alpha)}{(z-\Lm_a)},\\
G(z) &=
{\half}\frac{(z^2-2x^{-2}z+1)(z^2+\alpha^2-2y^2z)}{(z^2-1)(z^2-\alpha^2)}\sum_{a=1}^{2n}\frac1{z-\Lm_a},
\ea
\ee
observe that
\be
 V= \sum_{a=1}^{2n} \rez_{z=\Lm_a}\Bigl(\Psi(z)dz\Bigr).
 \ee
As
$\Psi(z) dz$
extends to a meromorphic 1-form on the Riemann sphere $\overline{\CC}$,
 the sum of its residues  over the poles in $\overline{\CC}$ is zero.
In addition to $z= \Lambda_a$ for $a=1,\dots, 2n$,
 $\Psi(z) dz$ possesses poles at $z= \pm 1, \pm \alpha,\infty$.
All residues
can be calculated straightforwardly.
In this way, using also the substitutions $\Lambda_j=e^{2\lambda_j}$, (\ref{xyalphsubs}) and elementary hyperbolic identities
like
$\sinh(\nu + \mu) \sinh(\nu - \mu) =  \sinh^2\nu - \sinh^2\mu$,
we obtain formula (\ref{pot}) for $V$.

To finish the derivation, we first rewrite (\ref{modtwsol}) as
\be\label{twk1}
\vert\tilde w_k\vert^2 =e^{-\mu}\left(e^{2 \q_k} - y^2\right) \frac{\sinh(\mu)}{\sinh(2\q_k)}
\prod_{\substack{i=1\\ (i\neq k)}}^{n}\left(\frac{\sinh(\q_k+\q_i+\mu)\sinh(\q_k-\q_i+\mu)}
{{\sinh(\q_k-\q_i)}\sinh(\q_k+\q_i)}\right)
\ee
 and
\be\label{twk2}
\vert\tilde w_{n+k}\vert^2 =e^{-\mu}\left(y^2 - e^{-2 \q_k}\right) \frac{\sinh(\mu)}{\sinh(2\q_k)}
\prod_{\substack{i=1\\ (i\neq k)}}^{n} \left(\frac{\sinh(\q_k+\q_i-\mu)\sinh(\q_k-\q_i-\mu)}
{{\sinh(\q_k-\q_i)}\sinh(\q_k+\q_i)}\right)
\ee
for $k=1,\dots, n$.
Substituting these in the second  term of (\ref{rawham})  and using again (\ref{xyalphsubs}) leads to
the claimed formula
(\ref{distilham}) for $\Phi_1^\red$.
\end{proof}

\section{Discussion}\label{discussion}

The Heisenberg double $\cM$ of the Poisson Lie group $K=\SU(2n)$, equipped
with the   Abelian Poisson algebras generated by $\{F_l\}$ and $\{\Phi_l\}$ (\ref{FPhidefs}), permits
Hamiltonian reduction  by the constraint in (\ref{constraint}).
All the functions $F_l$ and $\Phi_l$ are invariant with respect to the symmetry group
$K_+\times K_+$, and thus $\{F_l\}$ and $\{\Phi_l\}$ descend to
Abelian Poisson algebras on the reduced phase space $\cM_\red$ (\ref{Mred}), where
they engender two Liouville integrable systems.
The present paper continues the line of research  started in \cite{M} and
further advanced in \cite{FG,FG1,M1}.
The aim of these studies  is to achieve detailed understanding
of the integrable systems defined by the collections of reduced Hamiltonians $\{ F_l^\red\}$
and $\{\Phi_l^\red\}$ as well as their
 analogues obtained by using $\SU(n,n)$ instead of $\SU(2n)$ in the
 decompositions (\ref{lrdecomp}),(\ref{rldecomp}).
The pertinent reductions admit two natural models  for the reduced phase space,
which are associated with two systems of Darboux coordinates
on (dense open submanifolds of) $\cM_\red$.
The Darboux coordinates emerge from the eigenvalues of two matrices
complemented by their respective canonical conjugates. In our setting
these two matrices are $\Omega$ and $L$ (\ref{OmLdefs}).
The coordinates based on diagonalization of $L$ were described in \cite{FG},
following \cite{M}.
Here, we have constructed alternative Darboux coordinates
utilizing the eigenvalues $\Lambda_j= e^{2\lambda_j}$ of $\Omega$.

The canonical conjugates of the variables $\lambda_j$ are angles $\theta_j$,
parametrizing an $n$-torus ${\TT}^n$, but so
far we have not specified the range of the
eigenvalue-parameters $\lambda_j$: it will be proved in [29] that their full range is the closure
of the  convex polyhedron
\be
\cD_+^\lambda =
\{ \lambda\in {\RR}^n\mid \lambda_1> \lambda_2 > \dots > \lambda_n >
\operatorname{max}(\vert v\vert, \vert u \vert),\,\, \lambda_{i} - \lambda_{i+1}> \mu,\,\, i=1,\ldots, n-1\},
\label{Z1}\ee
where $\mu$, $u$ and $v$ are the constants (\ref{xyalphsubs}) appearing in the definition of the
constraint (\ref{constraint}).
The restriction of $\lambda$ to the domain $\cD_+^\lambda$ is a consequence of the facts
that the variables $\Lambda_j=e^{2\lambda_j}$ satisfy (\ref{Lambdaord})
  and that the functions $\vert \tilde w_a\vert^2$ in (\ref{modtwsol}) cannot be negative.
Indeed, these functions, exhibited also in (\ref{twk1})-(\ref{twk2}), are all positive precisely on the domain (\ref{Z1}).

We have seen that the reduced Hamiltonian $\Phi_1^\red$
takes the interesting RSvD  form (\ref{distilham}) in terms
of the Darboux coordinates  attached to
$\cD_+^\lambda \times {\TT}^n = \{ (\lambda, e^{\ri \theta})\}$.
On the other hand, in these coordinates
the reduced Hamiltonians $F_l^\red$  depend only on $\lambda$, as given by (\ref{Flred}).
This means that $\lambda_j, \theta_j$ are action-angle variables
for the Liouville integrable system $\{F_l^\red\}$, and
the $\theta$-tori are just the Liouville tori.
The boundary of the polyhedron $\cD_+^\lambda$ actually corresponds
to lower-dimensional Liouville tori.

Now we recall the other system of Darboux coordinates, denoted $(\hat p, \hat q)$ in \cite{FG}. The
$\hat q_j$ are  angles,  whereas the $\hat p_j$ are related
to the parameters $q_j$ of the generalized Cartan decomposition of
$k\in K$ utilized to obtain the formula (\ref{bLform}). Concretely \cite{M,FG}, we have
\be
e^{\hat p_j} = \sin(q_j).
\label{Z5}\ee
These variables encode the eigenvalues of
$L= k^\dagger I k I$ since $L$ is conjugate to the matrix
\be
\bpm\cos(2 q)& \ri\sin (2 q)\\ \ri\sin(2q)&\cos(2q)\epm,
\quad q =\diag(q_1,\dots, q_n).
\label{Z6}\ee
The range of the variables $\hat p_j$ can be shown  \cite{FG}\footnote{In this reference
the unnecessary assumption $v>u$ was made.}
to be the closure of the domain
\be
\cD_+^{\hat p} = \{\hat p\in{\RR}^n\,\mid
 \hat p_1 < \operatorname{min}(0,v-u),\,\,\, \hat p_j-\hat p_{j+1}>\mu\ (j=1,\dots,n-1)\}.
\label{Z7}\ee
The pair $(\hat p, e^{\ri \hat q})$ filling the domain $\cD_+^{\hat p} \times {\TT}^n$
yields Darboux coordinates on a dense open subset of $\cM_\red$,
and in these coordinates the Hamiltonians $\Phi_l^\red$ become trivial, while
$F_1^\red$ gives an interesting Hamiltonian of RSvD type. Specifically, one obtains
\be
\Phi_l^\red = \frac{1}{l}\sum_{j=1}^n \cos (2 l q_j(\hat p)),
\label{Z9}\ee
referring to (\ref{Z5}), and
\be
F_1^\red =U(\hat p) - \sum_{j=1}^n\cos(\hat q_j) U_1(\hat p_j)^{\tfrac{1}{2}}
\prod_{\substack{k=1\\(k\neq j)}}^n
\bigg[1-\frac{\sinh^2(\mu)}{\sinh^2(\hat p_j-\hat p_k)}
\bigg]^{\tfrac{1}{2}}
\label{Z9}\ee
with
\be
U(\hat p)=\frac{e^{-2u}+e^{2v}}{2}\sum_{j=1}^ne^{-2\hat p_j},\quad
U_1(\hat p_j) = \big[1-(1+e^{2(v-u)})e^{-2\hat p_j}
+e^{2(v-u)}e^{-4\hat p_j}\big].
\ee
Hence $\hat p_j, \hat q_j$ are action-angle
variables for the Liouville integrable system $\{\Phi_l^\red\}$, and the $\hat p_j$ serve
also as position variables for $F_1^\red$ (\ref{Z9}).
Incidentally, it is  manifest  from the identity
\be
U_1(\hat p_j) = 4 e^{v-u} e^{-2\hat p_j} \sinh(\hat p_j) \sinh(\hat p_j + u -v)
\ee
 that the Hamiltonian $F_1^\red$ (\ref{Z9}) is real on the domain (\ref{Z7}), as it must be
 on account of its action-angle form (\ref{Flred}).

We conclude from the above that the Liouville integrable systems
$\{ F_l^\red\}$ and $\{\Phi_l^\red\}$  \textit{ are in action-angle duality}.
Indeed,  $F_1^\red$ takes the RSvD
form (\ref{Z9}) in terms of the action-angle variables of $\{\Phi_l^\red\}$, and
$\Phi_1^\red$ is given by the other RSvD type formula (\ref{distilham}) in terms of
the action-angle variables
of $\{F_l^\red\}$.

As was  mentioned in the Introduction, the first systematic investigation of action-angle duality relied
on direct methods \cite{SR88, RIMS95}. Since then,  the reduction interpretation of most
(although still not all) examples of Ruijsenaars have been found, and also several new  cases
of action-angle duality were unearthed utilizing this method;
see \cite{JHEP, FK1, P,FG0} and references therein.
The present paper should be seen as a contribution to the research goal
to describe dual pairs for all RSvD type systems in reduction terms.

Global properties of the reduced phase space (\ref{Mred}) and consequences of the
duality for the dynamics  will be studied in
our subsequent publication \cite{LONG}.
The relation of $F_1^\red$ (\ref{Z9}) to the five-parameter family of RSvD Hamiltonians \cite{vD1}
was described in \cite{FG}, and in \cite{LONG}  we will also present such a
connection  for $\Phi_1^\red$ (\ref{distilham}).
We here only note (see Appendix A) that $\Phi_1^\red$
is a deformation of the action-angle dual of the trigonometric $\BC_n$
Sutherland Hamiltonian, as must be the case since $F_1^\red$ can be viewed as a
deformation of the latter \cite{M,FG}.

We wish to point out that their  reduction origin naturally associates Lax matrices
to the models obtained, basically because $\Omega$ and  $L$ (\ref{OmLdefs}) generate the commuting
Hamiltonians ({\ref{FPhidefs})  before reduction.  Recently there appeared new results
about Lax matrices for certain hyperbolic RSvD models \cite{PG}, and it  would be interesting to compare
those with the Lax matrices that arise in our setting.

We finally remark that the quantum mechanical (bispectral)  analogue of
our dual pair should be understood.
The recent paper by van Diejen and  Emsiz \cite{vDE} is certainly
relevant for finding the answer to this question.
We hope that our investigations will be developed in several directions in the future,
including bispectral aspects withal.

\bigskip\bigskip \noindent\bf Acknowledgements. \rm
This work was supported in part by the Hungarian Scientific Research
Fund (OTKA) under the grant K-111697.

\newpage
\renewcommand{\thesection}{\Alph{section}}
\setcounter{section}{0}
\section{Connection to the dual of the $\BC_n$ Sutherland model}

In this appendix we present the ``cotangent bundle limit'' of the Hamiltonian
$\Phi_1^\red$ (\ref{distilham}).
We find it convenient to introduce  the notation
\be
H_1(\lambda, \theta; u, v, \mu) := \Phi_1^\red(\lambda, \theta)
\ee
containing the coupling parameters $u,v,\mu$ as given in (\ref{distilham}).
Let us now take any positive parameter $r$ and consider the one-parameter family
of Hamiltonians
\be
H_r(\lambda, \theta; u,v,\mu) := H_1(r \lambda, \theta; r u, rv, r\mu),
\ee
which are defined, for any $r>0$,  on the same domain $\cD_+^\lambda \times {\TT}^n$ (\ref{Z1})
as $H_1$.
It is easy to check that $H_r$ has a limit as $r\to 0$.
Indeed, we obtain
\be
\lim_{r\to 0} H_r(\lambda, \theta; u,v,\mu) = H_0(\lambda, \theta; u, v,\mu)
\ee
with
\be\label{limitPhi1}
\ba
 H_0 (\lambda,\theta; u,v,\mu)&=
V_0(\lambda;u,v,\mu) +  \sum_{k=1}^n \cos(\theta_k)
\left[1 - \frac{v^2}{\lambda_k^2}\right]^{1/2} \left[1 - \frac{u^2}{\lambda_k^2} \right]^{1/2}\\
&\qquad\times
\prod_{\substack{l=1\\(l\neq k)}}^n \left[1 - \frac{\mu^2}{(\lambda_k - \lambda_l)^2}\right]^{1/2}
\left[1 - \frac{\mu^2}{(\lambda_k + \lambda_l)^2}\right]^{1/2}
\ea
\ee
where
\be
V_0(\lambda; u, v,\mu) = \frac{uv}{\mu^2}
\prod_{k=1}^n\left[1 - \frac{\mu^2}{\lambda_k^2} \right]
-\frac{uv}{\mu^2}.
\label{pot0}
\ee
The limiting Hamiltonian $H_0$  can be recognised as the action-angle dual of the standard trigonometric
$\BC_n$ Hamiltonian.   The latter was derived  in \cite{FG0}
 by reduction of the cotangent bundle of $T^* \U(2n)$,
and was denoted there by $\tilde H^0$.  Concretely, the correspondence with the notations used in equation (1.4) of
\cite{FG0}  is
\be
H_0(\lambda, \theta; u,v,2\mu) = \tilde H^0(\lambda, \vartheta; \kappa, \nu, \mu)
\ee
 under the substitutions
\be
 u \to - \kappa,  \quad v \to \nu, \quad   \theta \to \vartheta.
\ee
We note for completeness that \cite{FG0} adopted the inessential condition $\nu > \vert \kappa\vert \geq 0$.


\begin{thebibliography}{5}

\bibitem{Cal}
F. Calogero,
\textit{Solution of the one-dimensional $N$-body problem with quadratic and/or
inversely quadratic pair potentials},
J. Math. Phys. {\bf 12}, 419-436 (1971)

\bibitem{Suth}
B.~Sutherland,
\textit{Exact results for a quantum many-body problem in one dimension},
Phys. Rev. A {\bf 4}, 2019-2021 (1971)

\bibitem{Mos}
J.~Moser,
\textit{Three integrable Hamiltonian systems connected with isospectral deformations},
Adv. Math. {\bf 16}, 197-220 (1975)


\bibitem{OP1}
M.A.~Olshanetsky and A.M.~Perelomov,
\textit{Classical integrable finite-dimensional systems related to Lie algebras}, Phys.
Rep. {\bf 71}, 313-400 (1981)

\bibitem{RS86}
S.N.M.~Ruijsenaars and H.~Schneider,
\textit{A new class of integrable systems and its relation to solitons},
Ann. Phys. {\bf 170}, 370-405 (1986)

\bibitem{vD1}
J.F.~van Diejen,
\textit{Deformations of Calogero-Moser systems},
Theor. Math. Phys. {\bf 99}, 549-554 (1994); {\tt arXiv:solv-int/9310001}

\bibitem{N}
N.~Nekrasov,
\textit{Infinite-dimensional algebras, many-body systems and gauge theories},
pp.~263-299 in: Moscow Seminar in Mathematical Physics, AMS Transl. Ser. 2, Vol.~{\bf 191},
American Mathematical Society, 1999

\bibitem{Banff}
S.N.M.~Ruijsenaars, \textit{Systems of Calogero-Moser type}, pp. 251-352
in: Proceedings of the 1994 CRM-Banff Summer School Particles and Fields,
Springer, 1999

\bibitem{vDV}
J.F. van Diejen and L. Vinet (Editors), Calogero-Moser-Sutherland Models,
Springer, 2000

\bibitem{SuthR}
B.~Sutherland, Beautiful Models, Word Scientific, 2004

\bibitem{PolR}
A.P.~Polychronakos,
\textit{Physics and mathematics of Calogero particles},
J. Phys. A: Math. Gen. {\bf 39}, 12793-12845 (2006);
{\tt arXiv:hep-th/0607033}

\bibitem{EtiR}
 P.~Etingof,
Calogero-Moser Systems and Representation Theory,
European Mathematical Society, 2007

\bibitem{SR88}
S.N.M.~Ruijsenaars,
\textit{Action-angle maps and scattering theory for some finite-dimensional
integrable systems. I. The pure soliton case},
Commun. Math. Phys. {\bf 115}, 127-165 (1988)

\bibitem{RIMS95}
S.N.M.~Ruijsenaars,
\textit{Action-angle maps and scattering theory for some finite-dimensional
integrable systems III. Sutherland type systems and their duals},
Publ. RIMS {\bf 31}, 247-353 (1995)


\bibitem{KKS}
D.~Kazhdan, B.~Kostant and S.~Sternberg,
\textit{Hamiltonian group actions and dynamical systems of Calogero type},
Comm. Pure Appl. Math. {\bf XXXI}, 481-507 (1978)

\bibitem{JHEP}
V.~Fock,  A.~Gorsky, N.~Nekrasov and V.~Rubtsov,
\textit{Duality in integrable systems and gauge theories},
JHEP {\bf 07}, 028 (2000); {\tt arXiv:hep-th/9906235}

\bibitem{FK1}
L.~Feh\'er and C.~Klim\v c\'ik,
\textit{Poisson-Lie interpretation of trigonometric Ruijsenaars duality},
Commun. Math. Phys. {\bf 301}, 55-104 (2011); {\tt arXiv:0906.4198 [math-ph]}

\bibitem{M}
I.~Marshall,
\textit{A new model in the Calogero-Ruijsenaars family},
Commun. Math. Phys. {\bf 338}, 563-587 (2015); {\tt arXiv:1311.4641 [math-ph]}


\bibitem{FG}
L.~Feh\'er and T.F.~G\"orbe,
\textit{On a Poisson-Lie deformation of the BC(n) Sutherland system},
Nucl. Phys. B {\bf 901}, 85-114 (2015); {\tt arXiv:1508.04991 [math-ph]}

\bibitem{STS}
M.A.~Semenov-Tian-Shansky,
\textit{Dressing transformations and Poisson groups actions},
Publ. RIMS {\bf 21}, 1237-1260 (1985)

\bibitem{FG1}
L.~Feh\'er and T.F.~G\"orbe,
\textit{The full phase space of a model in the Calogero-Ruijsenaars family}, J. Geom. Phys.
{\bf 115},  139-149 (2017); {\tt arXiv:1603.02877 [math-ph]}

\bibitem{M1}
I.~Marshall,
\textit{ Spectral parameter dependent Lax pairs for systems of Calogero-Moser type},
Lett. Math. Phys. {\bf 107}, 619-642 (2017)


\bibitem{FP}
L.~Feh\'er and B.G. Pusztai,
\textit{A class of Calogero type reductions of free motion on a simple Lie group},
Lett. Math. Phys. {\bf 79}, 263-277 (2007); {\tt arXiv:math-ph/0609085}

\bibitem{P}
B.G.~Pusztai,
\textit{The hyperbolic BC(n) Sutherland and the rational BC(n)
Ruijsenaars-Schneider-van Diejen models: Lax matrices and duality},
Nucl. Phys. B  {\bf 856}, 528-551 (2012); {\tt arXiv:1109.0446 [math-ph]}

\bibitem{FG0}
L.~Feh\'er and T.F.~G\"orbe,
\textit{Duality between the trigonometric $BC_n$ Sutherland system and a completed
rational Ruijsenaars-Schneider-van Diejen system},
J. Math. Phys. {\bf 55}, 102704 (2014); {\tt arXiv:1407.2057 [math-ph]}

\bibitem{W}
G.~Wilson,
\textit{Collisions of Calogero-Moser particles and an adelic Grassmannian (with an
Appendix by I. G. Macdonald)}, Invent. Math. {\bf 133}, 1–41 (1998)



\bibitem{STSlectures}
M.A.~Semenov-Tian-Shansky,
\textit{Integrable systems: an r-matrix approach},  Kyoto preprint RIMS-1650, 2008,
www.kurims.kyoto-u.ac.jp/preprint/file/RIMS1650.pdf


\bibitem{Lu}
J-H.~Lu,
\textit{Momentum mappings and reduction of Poisson actions},
pp.~209-226 in: Symplectic Geometry, Groupoids, and Integrable Systems,
 Springer, 1991

\bibitem{LONG}
L.~Feh\'er and I. Marshall, preprint in preparation

\bibitem{PG}
B.G.~Pusztai and T.F.~G\"orbe,
\textit{Lax representation of the hyperbolic van Diejen dynamics with two coupling parameters},
to appear in Commun. Math. Phys.,  {\tt arXiv:1603.06710 [math-ph]}


\bibitem{vDE}
J.F.~van Diejen and E.~Emsiz,
\textit{Spectrum and eigenfunctions of the lattice hyperbolic Ruijsenaars-Schneider
system with exponential Morse term},
Annales Henri Poincar\'e {\bf 17}, 1615-1629 (2016); {\tt arXiv:1508.03829 [math-ph]}


\end{thebibliography}
\end{document}